\documentclass[11pt]{article}

\usepackage{graphicx}
\usepackage{amsmath, amssymb, amsthm, amsfonts}
\usepackage{fullpage}
\usepackage{mathptmx}
\usepackage{latexsym}
\usepackage[scaled]{helvet}
\usepackage{color}
\usepackage{subfigure}

\newcommand{\eat}[1]{}

\newcommand{\opt}{{\text{\sc opt}}}
\newcommand{\pr}{{\mathbb{P}}}

\newtheorem{theorem}{Theorem}
\newtheorem{lemma}{Lemma}

\title{Online Load Balancing on Unrelated Machines with Startup Costs\thanks{Part
of this work was done while the first author was visiting and the
second author was an intern at Microsoft Research, Redmond, WA 98052.}}
\author{Yossi Azar\thanks{
Blatavnik School of Computer Science,
Tel-Aviv University, Tel-Aviv 69978, Israel. Email: {\tt azar@tau.ac.il}.}
\and Debmalya Panigrahi\thanks{
Computer Science and Artificial Intelligence Laboratory,
Massachusetts Institute of Technology,
Cambridge, MA 02139, USA. Email: {\tt debmalya@mit.edu}.}}
\date{}

\begin{document}

\maketitle

\thispagestyle{empty}

\begin{abstract}
Motivated by applications in energy-efficient scheduling in data 
centers, Khuller, Li, and Saha introduced the {\em machine activation} problem 
as a generalization of the classical optimization problems of minimum set cover 
and minimum makespan scheduling on parallel machines. 
In this problem, a set of $n$ jobs
have to be distributed among a set of $m$ (unrelated) machines, given the
processing time of each job on each machine. Additionally, each machine incurs
a startup cost if at least one job is assigned to it. 
The goal is to produce a schedule of minimum total
startup cost subject to a constraint $\bf L$ on its makespan.
While Khuller~{\em et al} considered the offline version of this problem, 
a typical scenario in scheduling is one where jobs arrive online and 
have to be assigned to a machine immediately on arrival. 
We give an $(O(\log (mn)\log m), O(\log m))$-competitive randomized online algorithm
for this problem, i.e. the schedule produced by our algorithm has 
a makespan of $O({\bf L} \log m)$ with high probability, and a total expected 
startup cost of 
$O(\log (mn)\log m)$ times that of an optimal offline schedule with makespan 
$\bf L$. Our algorithm is almost optimal since it follows from previous results
that the two approximation factors cannot be improved to $o(\log m\log n)$ 
(under standard complexity assumptions) and $o(\log m)$ respectively. 

Our algorithms use the online primal dual framework introduced by 
Alon~{\em et al} for the online set cover problem, and subsequently developed 
further by Buchbinder, Naor and co-authors in various papers. To the best of 
our knowledge, all previous applications of this framework have been to 
linear programs (LPs) with either packing or covering constraints. One 
novelty of our application is that we use this framework for a mixed LP that 
has both covering and packing constraints. We combine the packing constraint with
the objective function to design a potential function on the machines that is exponential
in the current load of the machine and linear in the cost of the machine.
Then, we create a dynamic order of machines based on this potential function and
assign larger fractions of the job to machines that appear earlier in this 
order. This allocation is somewhat unusual in that
the increase in load on a machine is inverse in the 
value of this potential function itself, i.e. inverse exponential in the
current load on the machine. Finally, we show that we can round
this fractional solution online using a randomized algorithm. 
We hope that the algorithmic 
techniques developed in this paper to simultaneously handle packing and
covering constraints will be useful for solving other online optimization 
problems as well. 

\end{abstract}

\clearpage

\setcounter{page}{1}

\section{Introduction}

In recent times, the emergence and widespread use of large-scale data centers with 
massive power requirements has elevated the problem of energy-efficient scheduling 
to one of paramount importance (see e.g.~\cite{BirmanCR09}). A natural strategy 
for achieving energy savings is that of {\em partial shutdown}, i.e. only 
a subset of machines/processors are active at any point of time. This immediately leads to 
the following scheduling question: {\em which set of machines should be activated 
to serve a given set of jobs?} Note that such a schedule must address twin
objectives:
\begin{itemize}
	\item The total cost (e.g. in terms of energy consumption) of all 
	machines used in the schedule must be small (thereby achieving energy
	efficiency).
	\item The sum of processing times of all jobs assigned to any machine
	must be small (thereby satisfying throughput requirements).
\end{itemize}
Motivated by this application, Khuller, Li, and Saha~\cite{KhullerLS10} 
introduced the {\em machine activation} problem that involves scheduling jobs 
to machines so as to minimize the cost of the machines used in the schedule, while
ensuring that the ``load'' on any machine is small.
Observe that treating each of these objectives individually leads to classical
problems in combinatorial optimization, namely {\em minimum set cover} and
{\em minimum makespan scheduling} on parallel machines, 
that have been extensively studied over
the last thirty years. The novelty of the algorithm proposed in 
\cite{KhullerLS10} for the machine activation problem lies in being able to
handle both objectives simultaneously. 

More formally, let $M$ be a set of $m$ machines and $J$ be a set of $n$ jobs, 
where the {\em processing time} of job $j$ on machine $i$ is $p_{ij} > 0$.
Further, machine $i$ has {\em startup cost} $c_i$. A {\em schedule}
is defined as an assignment $S: J\rightarrow M$ of jobs to machines; we denote
the set of jobs assigned to machine $i$ in schedule $S$ by $J^{(S)}_i$. 
The set of {\em active machines} $M^{(S)}_A$ in schedule $S$ are the machines 
to which at least one job has been assigned, i.e. 
$M^{(S)}_A = \{i\in M: J^{(S)}_i \not= \emptyset\}$, and the cost of schedule $S$ 
is defined as $\sum_{i\in A^{(S)}} c_i$. The {\em load} $\ell^{(S)}_i$ on machine $i$
is the sum of processing times of all jobs assigned to machine $i$, 
i.e. $\ell^{(S)}_i = \sum_{j\in J^{(S)}_i} p_{ij}$, and the {\em makespan}
$\ell^{(S)}_{\max}$ is the maximum load on a machine, i.e.
$\ell^{(S)}_{\max} = \max_{i\in M} \ell^{(S)}_i$. (Often, we will drop the 
superscript $(S)$ in the above notation if the schedule is clear from the context.) 
The objective of the machine activation problem is to obtain a schedule of minimum 
cost, subject to the constraint that its makespan is at most some given value $\bf L$.

In real-life scheduling tasks, the set of jobs is often not known in advance.
This has motivated extensive algorithmic research in online scheduling problems, 
where the  set of machines are available offline but the jobs appear online 
and have to be scheduled to a machine when they arrive. A natural and 
important question left open in \cite{KhullerLS10} was to obtain an algorithm for 
the machine activation problem in the online model. Here, the set of machines,
their individual startup costs, and the budget on the total startup cost of the 
machines activated by the schedule are known offline, but the jobs arrive online.
The processing time of a job on each machine is also revealed on arrival of the 
job. The goal is to assign the arriving job to a machine such that the cost
of the resulting schedule is minimized subject to the constraint that its makespan
is at most $\bf L$. We call this the {\em online machine activation problem}.

\noindent
\paragraph{Our Contributions.} 
Our main contribution is a randomized online algorithm 
for the machine activation problem with a bicriteria competitive ratio of 
$(O(\log (mn)\log m), O(\log m))$: suppose an offline 
optimal schedule for an instance of our problem has cost $B$ and makespan at most $\bf L$; 
then our online algorithm produces a schedule of expected
cost $O(B\log (mn)\log m)$ and makespan $O({\bf L}\log m)$
with high probability. 
\begin{theorem}
\label{thm:main}
There is a randomized online algorithm for the machine activation problem that
has a bicriteria competitive ratio of $(O(\log (mn)\log m), O(\log m))$.
\end{theorem}

In the minimum set cover problem, we are given a collection of subsets defined on 
a universe of elements. The goal is to select a minimum cost sub-collection
such that every element of the universe is in at least
one selected subset. If the elements appear online
and the current selection of subsets must cover every element that has
appeared thus far, a lower bound of $\Omega(\log m\log n)$ is 
known~\cite{Korman05}
for $m$ sets and $n$ elements, under standard complexity assumptions. 
Since the set cover problem is a special case of the machine activation problem
where the limit $\bf L$ on the makespan of the schedule is $\infty$, 
the competitive ratio in the cost of the schedule for any online algorithm 
for the machine activation problem must be $\Omega(\log m \log n)$.

On the other hand, the (minimum makespan) scheduling problem 
for unrelated parallel machines is 
defined as that of distributing $n$ jobs among $m$ machines so as to minimize the 
makespan of the schedule (machines do not have cost). It can be shown using standard 
techniques that an online algorithm that produces a schedule of makespan at most 
$\alpha{\bf L}$ for the online machine activation 
problem can be used to obtain an $O(\alpha)$-competitive algorithm for the online
scheduling problem on unrelated parallel machines. 
It is well-known~\cite{AzarNR95} that the competitive ratio of any algorithm 
for the latter problem is $\Omega(\log m)$; therefore, this
lower bound also holds for the online machine activation problem. 
\begin{theorem}
\label{thm:lowerbound}
No algorithm for the online machine activation problem can have a competitive
ratio of $o(\log m)$ in the makespan. Further, under standard complexity 
assumptions, no algorithm for the online machine activation problem can have 
an approximation factor of $o(\log m \log n)$ in the cost of the schedule.
\end{theorem}
\noindent
\paragraph{Our Techniques.}
Our algorithm draws inspiration from the techniques used to solve the online 
versions of the set cover problem and the scheduling problem for unrelated 
parallel machines. So, let us first summarize the key ideas involved in these
two algorithms. For the latter problem, Aspnes {\em et al}~\cite{AspnesAFPW97} 
gave an elegant solution based on the following exponential potential function: 
if the current load on machine $i$ is
$\ell_i$, then its potential is  $a^{\ell_i}$ for some constant $a$. 
The algorithm assigns the arriving job to a machine that suffers the minimum
increase of potential, i.e. to machine 
$i = \arg\min_{i\in M} (a^{\ell_i+p_{ij}} - a^{\ell_i})$. Observe that if all
processing times are scaled down sufficiently and $a$ is small enough, then
$a^{\ell_i + p_{ij}} - a^{\ell_i} \simeq a^{\ell_i} (a-1) p_{ij}$, i.e. the 
increase in potential function is linear in the processing time but exponential
in the current load on the machine. Therefore, the algorithm favors 
lightly loaded machines in preference to those offering low processing times.
It was shown in \cite{AspnesAFPW97} that this strong ``bias'' for lightly loaded 
machines ensures that the sum of the potentials of all machines is at most 
$O(m)$ times that in an optimal offline schedule, thereby leading to a competitive 
ratio of $O(\log m)$ on the makespan of the schedule. 

For the online set cover problem, Alon {\em et al}~\cite{AlonAABN09} introduced 
a two-phase online primal dual framework that works as follows. In the first 
phase, the goal is to obtain a feasible fractional solution to the online 
instance of the problem. In each step of this phase, in response to a new 
constraint (i.e. a new element in the online set cover problem) 
that arrives online, the fractional solution is updated to preserve 
feasibility while ensuring that the cost incurred can be accounted for by a 
suitably updated dual solution.\footnote{The dual was not used explicitly in Alon 
{\em et al}'s original analysis. In fact, we will also not use the dual explicitly
even though our algorithm can also be analyzed via the dual.}
In the second phase, the fractional solution is rounded {\em online} to obtain 
an integer solution. It is important to note that while the two phases are presented
sequentially for clarity, the algorithm needs to operate both phases (the fractional
updates followed by the rounding) in response
to response to the arrival of a new constraint. 
These two phases aim to notionally distinguish between 
the information-theoretic aspect of the online problem and the computational aspect
of rounding a fractional solution. In fact, Alon~{\em et al} showed 
for $m$ sets and $n$ elements, this two-phase framework can be used to 
design an algorithm that has a competitive ratio of $O(\log m\log n)$, where the 
$O(\log m)$ factor arises in the first phase due to information-theoretic
limitations of the online algorithm, and the $O(\log n)$ factor arises in the
second phase due to computational limitations encapsulated by the integrality
gap of the linear program (LP) for set cover.

\begin{figure}[!htb]
	\centering
	\small
	Minimize $\quad\sum_{i\in M} c_i x_i\quad$ subject to
	\begin{eqnarray}
		\label{eqn:packing}		\sum_{j\in J} p_{ij} y_{ij} & \leq &  x_i {\bf L} \quad \forall i\in M \\
		\label{eqn:fraction}		y_{ij} & \leq & x_i\quad \forall~i\in M,~j\in J \\
		\label{eqn:covering}		\sum_{i\in M} y_{ij} & \geq & 1\quad \forall~j\in J \\
		\label{eqn:range1}		x_i & \in & \{0, 1\} \quad \forall i\in M \\
		\label{eqn:range2}		y_{ij} & \in & \{0, 1\}\quad \forall i\in M,~j\in J 
	\end{eqnarray}
	\caption{\small The integer scheduling linear program (or ISLP). In the fractional scheduling linear
		program (or FSLP), Eqns.~\ref{eqn:range1} and \ref{eqn:range2} are relaxed to $0 \leq x_i \leq 1$
		for all machines $i\in M$, and $0\leq y_{ij}\leq 1$ for all machines $i\in M$ and jobs $j\in J$, 
		respectively.}
	\label{fig:lp}
\end{figure}

This two-phase primal dual framework has since been extensively  
used for various online problems (see \cite{BuchbinderN09b} for a survey);
our algorithm also uses this framework. However, to the best of our knowledge,
whereas all previous applications of the framework have been to LPs that have 
exclusively covering or packing constraints, we apply the framework to a mixed LP. 
Consider the integer LP formulation of our problem 
given in Fig.~\ref{fig:lp} (we call this the {\em integer scheduling LP} or ISLP). 
The variable $x_i$ is 1 iff machine $i$ is active, 
and $y_{ij}$ is 1 iff job $j$ is assigned to machine $i$. In the fractional
relaxation (which we call the {\em fractional scheduling LP} or FSLP), 
these variables are constrained to be in the range $[0, 1]$ 
instead of Eqns.~\ref{eqn:range1} and \ref{eqn:range2}. 
Note that we have both covering 
(Eqn.~\ref{eqn:covering}) and packing (Eqn.~\ref{eqn:packing}) constraints
in the ISLP/FSLP.

We interpret the online set cover algorithm as one that 
maintains a bound on a potential function that is linear in the cost of the 
active machines, whereas the online scheduling algorithm translates its 
objective of minimizing makespan into maintaining a bound on the value of 
a potential function that is exponential in the load on each machine. We 
design a potential function that combines both objectives: it is linear in
the startup cost and exponential in the load on a machine. Our fractional
algorithm preferentially assigns (fractions of) jobs to machines in a way that 
leads to a small increase in this potential. It should be noted that even 
if we use this potential function, we cannot afford to simply assign
the entire job  to the machine that would suffer the minimum increase in 
potential---such a 
greedy strategy can be easily shown to fail even for the special case of the
online set cover problem. Instead, our algorithm creates a dynamic list of 
machines in increasing order of its change in potential if the current job 
were assigned to it, and then assigns larger fractions of the job to 
machines that appear earlier in the order. 
This allocation is also somewhat unusual in that
the increase in load on a machine is inverse in the 
value of this potential function itself, i.e. inverse exponential in the
current load on the machine. For some of the machines in the
order, this might also involve increasing the fraction to which the machine
is active, i.e. increasing the value of $x_i$, if the assignment violates 
Eqn.~\ref{eqn:packing} or \ref{eqn:fraction}. The analysis of the fractional
algorithm involves proving a bound on the value of the cumulative
potential function over all the machines. Depending on the behavior of a
fixed optimal offline solution (which is unknown to the algorithm and is 
used only for analysis), we classify jobs into three different categories, and
prove a bound on the total increase in potential due to jobs in each individual
category using three different techniques:
\begin{itemize}
\item The increase of potential for jobs in category 1 is charged globally to 
the offline optimal solution using a primal dual argument (without introducing
the dual solution explicitly) similar in spirit to that used by Alon {\em et al} 
in~\cite{AlonAABN09} for the online set cover problem.
\item We give a bound on the increase in potential for each individual job in 
category 2 by showing that the ratio of increase in potential to the fraction
of job assigned is bounded.
\item We give a global bound on the increase of potential for jobs in category
3 by using a recursive argument which is similar in spirit 
to the one used in the online scheduling
algorithm for unrelated machines by Aspnes {\em et al}~\cite{AspnesAFPW97}.
\end{itemize}
The novelty of our analysis lies in being able to seamlessly combine, and
non-trivially extend, the disparate 
techniques from \cite{AlonAABN09} and \cite{AspnesAFPW97}.
Ultimately, we show that the potential of the fractional schedule 
produced by our algorithm is $O(m\log m)$, where the cost and makespan
of the offline optimal solution are respectively $\Omega(m)$ and 1 by 
a suitable initial scaling. We hope that the algorithmic 
techniques developed in this paper to simultaneously handle packing and
covering constraints will be useful for solving other online optimization 
problems that can be expressed as mixed LPs as well. 

In the second phase of the algorithm, we use an online randomized rounding scheme 
to obtain an integer solution. To ensure that the expected cost of 
the integer schedule is bounded by that of the fractional schedule, each machine $i$ 
is activated with probability proportional to $x_i$. This is implemented online using 
standard techniques. The more challenging aspect of the rounding is the actual 
scheduling of jobs to active machines. The natural approach would be to assign 
job $j$ to machine $i$ with probability $y_{ij}$. 
Translated to conditional probabilities, this implies
that job $j$ should be assigned to machine $i$ with probability $z_{ij} = y_{ij}/x_{ij}$
if machine $i$ is active, where $x_{ij}$ is the value of $x_i$ at the end of the
update to the fractional solution for job $j$. This immediately implies that the expected
load on a machine in the integer schedule is at most that in the fractional schedule. 
However, our goal is to obtain a bound
on the {\bf makespan} of the integer solution; in fact, since the events of 
jobs being assigned to a fixed machine are positively
correlated, a bound on the expected load does not immediately
yield a concentration bound on the load. Instead, we show that even if job $j$ 
were to be assigned to machine $i$ {\em unconditionally} with probability $z_{ij}$,
the expected load on machine $i$ given by $\sum_{j\in J} z_{ij}$ would be small. 
This overcomes the problem of positive correlation mentioned above since the events are 
no longer conditioned on machine $i$ being active. We now derive
concentration bounds on the load on a machine, which translates to a bound on the 
makespan of the integer schedule thereby proving Theorem~\ref{thm:main}.

\noindent
\paragraph{Previous Work.} 
Many variants of the machine scheduling (or {\em load balancing}) problem have been 
extensively studied in the literature. Perhaps the most celebrated result in 
scheduling theory is a 2-approximation for the offline minimum makespan scheduling problem
for unrelated machines due to Lenstra, Shmoys and Tardos~\cite{LenstraST90},
which was later simplified by Shmoys and Tardos~\cite{ShmoysT93}.
Rather surprisingly, this algorithm continues to offer the best competitive ratio for
this problem (and even for several natural special cases such as the restricted
assignment problem) even after more than two decades
of research. 
In the online setting, Graham~\cite{Graham66, Graham69} showed that the natural greedy
heuristic achieves
a competitive ratio of $2-1/m$ for $m$ identical machines. The competitive ratio of 
this problem has been subsequently improved in a series of results 
(see e.g.~\cite{BartalFKV95} and subsequent improvements). 
For the more general restricted assignment problem
where the processing time of each job $j$ on any machine is either some value $p_j$ 
or $\infty$, an online algorithm having competitive ratio $O(\log m)$ was designed
by Azar, Naor and Rom~\cite{AzarNR95}. This algorithm was later generalized to 
the unrelated machines scenario by Aspnes {\em et al}~\cite{AspnesAFPW97} 
with the same competitive ratio.
Various other models and objectives have been considered for the load balancing 
problem; for a comprehensive survey, see \cite{Azar96} and \cite{Sgall96}. 
In particular, the machine activation problem was introduced by 
Khuller~{\em et al} in \cite{KhullerLS10}, where they gave an 
$O(2(1+1/\epsilon)(1 + \ln (n/OPT)), 2+\epsilon)$-approximation algorithm for
any $\epsilon > 0$. Recently, this result was extended by Khuller and
Li~\cite{LiK11} to a more general set of cost functions.


\section{The Fractional Algorithm}
\label{sec:fractional}

\begin{figure}[!htb]
	\small
	\centering
	Minimize $\quad\sum_{i\in M} c_i x_i\quad$ subject to
	\begin{eqnarray}
		\label{eqn:packingnew}	\sum_{j\in J} p_{ij} y_{ij} & \leq & 6 x_i {\bf L} \\
		\label{eqn:fractionnew}		y_{ij} & \leq & 2 x_i\quad \forall~i\in M,~j\in J \\
		\label{eqn:coveringnew}		\sum_{i\in M} y_{ij} & \geq & 1\quad \forall~j\in J \\
		\label{eqn:range1new}		0\quad \leq & x_i & \leq \quad 1 \quad \forall i\in M \\
		\label{eqn:range2new}		0\quad \leq & y_{ij} & \leq \quad 1 \quad \forall i\in M,~j\in J 
	\end{eqnarray}
	\caption{\small The relaxed fractional scheduling linear program (or RFSLP). 
		Eqn~\ref{eqn:packingnew} is enforced only for partially active machines $i$.
		(Note that for inactive machines, Eqns.~\ref{eqn:packingnew} and
		\ref{eqn:fractionnew} are identical.)}
	\label{fig:lp-relaxed}
\end{figure}

In this section, we will describe the online updates to the fractional
solution to maintain feasibility for FSLP on receiving a new job $j$. 
This involves updating the values of 
$y_{ij}$ (denoting the fraction of job $j$ assigned to machine $i$)
so as to satisfy Eqn.~\ref{eqn:covering}, and corresponding updates 
to the values of $x_i$ if Eqns.~\ref{eqn:packing} or \ref{eqn:fraction}
is violated. 
In fact, we relax the constraints in FSLP in two ways.
Let machine $i$ be said to be {\em inactive}, {\em partially active} or
{\em fully active} depending on whether $x_i = 0$, $0 < x_i < 1$ or
$x_i = 1$ respectively. First, for technical reasons, 
we relax Eqns.~\ref{eqn:packing} and \ref{eqn:fraction} 
to Eqns.~\ref{eqn:packingnew} and \ref{eqn:fractionnew} respectively
(see Fig.~\ref{fig:lp-relaxed}). Further, we 
enforce Eqn.~\ref{eqn:packing} only if $x_i < 1$, i.e. if machine 
$i$ is not fully active. The load on a fully active machine will
be bounded separately in the analysis.
We call this the {\em relaxed fractional scheduling LP} or RFSLP.

Before describing these updates, let us set up some 
conventions that we will use throughout the paper.
We divide all processing times by $\bf L$ at the outset;
this allows us to assume that the makespan of the optimal solution is 1.
We also assume that we know the value $\alpha$ of the optimal 
offline (integer) solution, i.e. the minimum startup cost of an offline 
assignment of jobs to machines that has makespan at most 1. 
This is also without loss of generality because we can guess the 
value of the optimal solution, doubling our guess whenever the 
current algorithmic solution exceeds the cost bounds that we are 
going to prove (thereby implying that our current guess is too small). 
In the following discussion, it is sufficient to
know the value of $\alpha$ up to a multiplicative 
factor of 2, but for simplicity
of presentation, we will assume that we know it exactly.

Our algorithm uses the value of $n$, which is the total number of
jobs that arrive online. If this value is not known offline, each job 
estimates $n$ by assuming that it is the last job. We can show that 
using such estimates for $n$ incurs a small {\em additive} factor of 
$O(\log \log n)$ in the makespan, and an {\em additive} factor of 
$O(\log n\log (mn))$ in the cost of the schedule. 
However, for simplicity, the rest of the paper assumes that the value 
of $n$ is known offline.

We define the {\em virtual cost} of job $j$ on machine $i$ as 
\begin{equation*}
\eta_i(j) = \left\{ \begin{array}{ll}
c_i a^{\ell_i-1} p_{ij}, & \text{if machine $i$ is fully active, i.e. } x_i = 1 \\
c_i p_{ij}, & \text{otherwise}
\end{array}\right.
\end{equation*}
where $a$ is a constant that we will fix later. 
(Recall that $\ell_i$ represents the load on machine $i$, 
i.e. $\ell_i = \sum_{j\in J} p_{ij} y_{ij}$.)
Let $M(j)$ denote
an ordering of machines in non-decreasing order of virtual cost 
$\eta_i(j)$ for job $j$. 
Let $P(j)$ denote the maximal prefix of $M(j)$ such that 
$\sum_{i\in P(j)} x_i < 1$. (Note that $P(j)$ may be empty.)  
If $P(j)\not= M(j)$, then $k(j)$ 
denotes the first machine in $M(j)$ that is not in $P(j)$; 
$k(j)$ is undefined if $P(j) = M(j)$.

If $x_i$ is increased to $x_i + \Delta x_i$ for a partially
active machine $i$, then
we say that the {\em effective capacity} created by this increase 
for job $j$ is $\min(2 x_i, 6 \Delta x_i/p_{ij})$.
Note that the effective capacity created by an increase in $x_i$ 
is a feasible increase in the value of $y_{ij}$ independent of the 
current load on machine $i$.

\subsection{The Algorithm}
The algorithm has two phases---an offline pre-processing phase, and
an online phase that (fractionally) schedules the arriving jobs.
 
\paragraph{Pre-processing.} We multiply the startup cost of every machine
by $\alpha/m$, and discard machines with startup cost greater than $m$
(after the scaling) at the outset. Further, for every machine $i$
whose startup cost is at most 1, we increase its
cost to 1, and initialize $x_i$ to 1. For all other
machines with $1 < c_i \leq m$, we initialize $x_i$ to $1/m$.
At the end of the pre-processing phase, we have the following properties:
\begin{itemize}
	\item The cost of an optimal solution is between $m$ and $2m$.
	\item The cost of every machine is between 1 and $m$.
	\item Every machine whose cost is 1 is fully active; all other
	machines have $x_i = 1/m$.
\end{itemize}


\paragraph{Online Algorithm.}
Suppose job $j$ arrives online. 
We increase $y_{ij}$s using the following rules repeatedly until 
Eqn.~\ref{eqn:coveringnew} is satisfied.
\begin{itemize}
\item {\bf Type A: $k(j)$ is undefined (i.e. $P(j) = M(j)$) or 
$x_{k(j)} < 1$ (i.e. machine $k(j)$ is not fully active).} 
We increase $x_i$ to 
$\min(x_i(1+1/c_i n), 1)$ for each machine 
$i\in P(j)$ (and also for $i= k(j)$ if it is defined), and correspondingly
increase $y_{ij}$ by the effective capacity created in each machine.
\item {\bf Type B: $x_{k(j)} = 1$ (i.e. machine $k(j)$ is fully active).}
We increase $x_i$ to 
$\min(x_i(1 + 1/c_i n), 1)$ for each machine 
$i\in P(j)$, and correspondingly increase $y_{ij}$ by the effective 
capacity created in each machine. Further,
we increase $y_{k(j)j}$ by 
$6/\eta_{k(j)}(j) n$. 
\end{itemize}

\subsection{Analysis}
Our goal is to show the following bounds
on the makespan and cost of the fractional schedule.
\begin{lemma}\label{lma:fractional-main}
The fractional schedule produced by the online algorithm satisfies 
$\sum_{j\in J} y_{ij} p_{ij} = O(\log m)$ 
for each machine $i$, and  
$\sum_{i\in M} c_i x_i = O(m\log m)$.
\end{lemma}
\noindent
We introduce a potential function $\phi_i$ for machine $i$ defined as 
\begin{equation*}
\phi_i = \left\{ \begin{array}{ll}
c_i a^{\ell_i-1}, & \text{if machine $i$ is fully active, i.e. $x_i = 1$} \\
c_i x_i, & \text{otherwise.}
\end{array}\right.
\end{equation*}
The cumulative potential function $\phi = \sum_{i\in M} \phi_i$.
Our goal will be to show that $\phi = O(m\log m)$. 
This will immediately imply Lemma~\ref{lma:fractional-main}.

We prove the bound on $\phi$ in three steps. First, we bound the increase of 
$\phi$ in the pre-processing phase; next, we bound the increase of $\phi$ in
each algorithmic step (of either type A or type B); and finally, we bound the 
total number of algorithmic steps. 

\noindent
{\bf Pre-processing.}
The next lemma bounds the increase of $\phi$ in the pre-processing phase.
\begin{lemma}
\label{lma:initial}
At the end of the pre-processing phase, $\phi \leq m$.
\end{lemma}
\begin{proof}
The startup cost of each machine that is fully active after pre-processing is 1;
on the other, every partially active machine $i$ has $c_i\leq m$ and $x_i = 1/m$
after pre-processing.
\end{proof}

\noindent
{\bf Single Algorithmic Step.}
Now, we bound the increase in $\phi$ due to a single algorithmic step of either type.
\begin{lemma}
\label{lma:singleA}
The increase in potential in a single algorithmic step of type A 
is at most $2/n$.
\end{lemma}
\begin{proof}
The total increase in potential in an algorithmic step of type A (due to increase
in the value of $x_i$ for machines $i\in P(j)\cup \{k(j)\}$) is 
at most 
$\sum_{i\in P(j)\cup \{k(j)\}} c_i (x_i/c_i n) = (\sum_{i\in P(j)\cup \{k(j)\}} x_i)/n < 2/n$.
\end{proof}
\noindent
\begin{lemma}
\label{lma:singleB}
For any constant $1 < a < 13/12$, the increase in potential in a single algorithmic 
step of type B is at most $2/n$.
\end{lemma}
\begin{proof}
The total increase in potential for machines $i\in P(j)$ due to an 
algorithmic step of type B is at most
$\sum_{i\in P(j)} c_i (x_i/c_i n) = (\sum_{i\in P(j)} x_i)/n < 1/n$.
The other source of increase in potential is the scheduling of 
a fraction of job $j$ to machine $k(j)$, due of which
the load on machine 
$k(j)$ increases by $6/c_{k(j)}a^{\ell_{k(j)}-1} n$. The resulting
increase of potential $\phi_{k(j)}$ is 
\begin{eqnarray*}
c_{k(j)} (a^{\ell_{k(j)} - 1 + 6/c_{k(j)}a^{\ell_{k(j)} - 1} n} - a^{\ell_{k(j)} - 1}) 
& = & c_{k(j)} a^{\ell_{k(j)} - 1} (a^{6/c_{k(j)}a^{\ell_{k(j) - 1}} n} - 1) \\
\quad = \quad c_{k(j)} a^{\ell_{k(j)} - 1} \left(\left(1 + (a-1)\right)^{6/c_{k(j)} a^{\ell_{k(j)} - 1} n} -  1 \right)
& < & c_{k(j)} a^{\ell_{k(j)} - 1} \cdot \frac{12(a-1)}{c_{k(j)} a^{\ell_{k(j)} - 1} n} 
\quad < \quad \frac{1}{n}.
\end{eqnarray*}
The penultimate inequality follows from the property that
$(1 + x)^{1/y} < e^{x/y} < 1 + 2x/y$, 
for any $y \geq x > 0$.
\end{proof}

\noindent
{\bf Number of Algorithmic Steps.}
We classify the algorithmic steps according to a fixed 
optimal offline (integer) schedule that we call {\sc opt}. Suppose {\sc opt}
assigns job $j$ to machine $\opt (j)$, and let $M_{\opt}$ denote the machines
that are active in the optimal offline schedule. The three categories are:
\begin{enumerate}
\item $\opt (j) \in P(j)$.
\item $\opt (j) \notin P(j)$ and $\opt (j)$ is partially active.
\item $\opt (j)$ is fully active.
\end{enumerate}
The next lemma bounds the total increase in potential due to algorithmic 
steps in the first category above.
\begin{lemma}
\label{lma:category1}
The total increase in potential due to all algorithmic steps in the first 
category is $O(m \log m)$.
\end{lemma}
\begin{proof}
In any algorithmic step of the first category, the value of $x_{\opt (j)}$
either increases to 
$x_{\opt (j)}\left(1 + \frac{1}{c_{\opt (j)} n}\right)$ or to 1. 
Since $x_i$ is initialized
to at least $1/m$ for each machine $i$ in the pre-processing phase, 
there are at most 
$m + \sum_{i\in M_{\opt}} c_i n\log m = O(m n\log m)$ 
algorithmic steps of the first category.
The lemma now follows from Lemmas~\ref{lma:singleA} and \ref{lma:singleB}.
\end{proof}
\noindent
The next lemma bounds the total increase in potential due to algorithmic 
steps in the second category.
\begin{lemma}
\label{lma:category2}
The total increase in potential due to all algorithmic steps in the second 
category is $O(m)$.
\end{lemma}
\begin{proof}
Consider the first 
$2 c_{\opt (j)} p_{\opt (j)j} n$ algorithmic steps in the second category for any 
particular job $j$. We have two cases: 
\begin{itemize}
\item {\bf Case~1}: these
algorithmic steps contain at least $c_{\opt (j)} p_{\opt (j)j} n$ steps of 
type B, or
\item {\bf Case~2}: these algorithmic steps contain at least
$c_{\opt (j)} p_{\opt (j)j} n$ steps of type A.
\end{itemize} 
In case~1, each algorithmic step of type B creates an effective capacity
of $6/\eta_{k(j)}(j) n$ in machine $k(j)$. Since in each such algorithmic
step, $\eta_{k(j)}(j) \geq \eta_{\opt (j)}(j) = c_{\opt (j)}p_{\opt (j)j}$,
the total effective capacity created by these algorithmic steps is 
at least 1.

In case~2, let $R(j)$ denote the set $P(j)\cup \{k(j)\}$ 
for the last of these algorithmic steps. (Note that since 
$\opt (j)\notin P(j)$, $k(j)$ is defined.) Further, let $x^{(1)}_i$ and
$x^{(2)}_i$ respectively denote 
the value of $x_i$ for machine $i$ before the first
algorithmic step for job $j$, and after the last
algorithmic step. 
For each machine $i\in R(j)$, $x_i$ has been increased in 
each of at least 
$c_{\opt (j)} p_{\opt (j)j} n$ algorithmic steps of type A. Thus, for each 
machine $i\in R(j)$,
\begin{equation}
\label{eqn:diff}
x^{(2)}_i 
\geq x^{(1)}_i \left(1 + \frac{1}{c_i n}\right)^{c_{\opt (j)} p_{\opt (j)j} n} 
= x^{(1)}_i \left(\left(1 + \frac{1}{c_i n}\right)^{c_i n}\right)^{\frac{c_{\opt (j)} p_{\opt (j)j}}{c_i}} 
\geq x^{(1)}_i 2^{\frac{c_{\opt (j)} p_{\opt (j)j}}{c_i}} 
\geq x^{(1)}_i 2^{p_{ij}}
\end{equation}
since $c_{\opt (j)} p_{\opt (j)j} \geq c_i p_{ij}$ for all machines $i\in R(j)$. The 
total effective capacity created in these steps is 
\begin{equation*}
\frac{6 (x^{(2)}_i - x^{(1)}_i)}{p_{ij}} 
\geq \frac{6 x^{(2)}_i (1 - 2^{-p_{ij}})}{p_{ij}}
\geq \frac{6 x^{(2)}_i}{2}
> x^{(2)}_i.
\end{equation*}
The first inequality follows from Eqn.~\ref{eqn:diff} while the 
second inequality follows from the observation that for any $z \leq 1$,
we have
$(1 - 2^{-z})/z \geq [1 - (1 - z + z^2/2)]/z = 1 - z/2 \geq 1/2$.
Hence, the effective capacity created by these algorithmic steps is at least
$\sum_{i\in R(j)} x^{(2)}_i \geq 1$. The lemma now follows from
Lemmas~\ref{lma:singleA} and \ref{lma:singleB} coupled with the fact that
$\sum_{j\in J} c_{\opt (j)} p_{\opt (j)j} \leq 2m$.
\end{proof}
\noindent
Finally, we bound the total increase of potential due to algorithmic 
steps in the third category.
\begin{lemma}
\label{lma:category3}
For any $1 < a < 13/12$,
the total increase in potential due to all algorithmic steps in the third 
category is at most $2 + (2/3)\sum_{i\in M_A} c_i a^{L_i - 1}$, 
where $L_i$ is the final load on machine $i$ in the schedule produced
by the algorithm and $M_A$ is the set of machines that are fully 
activated by the algorithm.
\end{lemma}
\begin{proof}
First,
we consider an algorithmic step of type A. In such a step, $k(j)$ must 
be defined since $\opt (j)\notin~P(j)$; thus,
$\sum_{i\in P(j)\cup \{k(j)\}} x_i \geq 1$. Further, for each machine
$i\in P(j)\cup \{k(j)\}$, we have 
$c_i p_{ij} = \eta_i(j) \leq \eta_{\opt (j)}(j)$.
Thus, the fraction of job $j$ assigned in this algorithmic step is at least
\begin{equation*}
\sum_{i\in P(j)\cup \{k(j)\}} \min\left(\frac{6 x_i}{c_i p_{ij} n}, 2x_i \right) \geq \min\left(\frac{3}{\eta_{\opt (j)}(j) n}, 1\right),
\end{equation*}
since $\sum_{i\in P(j)\cup\{k(j)\}} x_i \geq 1$.
We first consider the situation where the whole of job $j$ was assigned 
in this algorithmic step. In this case, the increase in potential due
to job $j$ is 
$\sum_{i\in P(j)\cup \{k(j)\}} (c_i x_i)/(c_i n) < 2/n$; cumulatively for
all jobs, such increases in potential add up to at most 2.
Otherwise, the sum of increase in $y_{ij}$ over all machines in 
this algorithmic step is at least $3/\eta_{\opt (j)j} n$. 
Now, consider an algorithmic step of type B.
Then the increase in $y_{k(j)j}$ is $6/\eta_{k(j)j} n \geq 6/\eta_{\opt (j)j} n$ 
since $\eta_{\opt (j)j} \geq \eta_{k(j)j}$. In either case, 
the total number of algorithmic steps in the third category
for job $j$ is at most $\eta_{\opt (j)j} n/3$. 

Note that $\eta_{\opt (j)j} \leq c_{\opt (j)} a^{L_{\opt (j)}-1} p_{\opt (j)j}$.
Therefore, summing over all jobs, the total increase in potential 
due to algorithmic steps in the third category is bounded by 
(using Lemmas~\ref{lma:singleA} and \ref{lma:singleB})
\begin{eqnarray*}
2 + \frac{2}{3} \sum_{j\in J} c_{\opt (j)} a^{L_{\opt (j)}-1} p_{\opt (j)j}
& = & 2 + \frac{2}{3} \sum_{i\in M_{\opt} \cap M_A} c_i a^{L_i-1} \left(\sum_{j: \opt (j) = i} p_{ij} \right) \\
\quad \leq \quad 2 + \frac{2}{3} \sum_{i\in M_{\opt} \cap M_A} c_i a^{L_i-1}
& \leq & 2 + \frac{2}{3} \sum_{i\in M_A} c_i a^{L_i-1},
\end{eqnarray*}
where the penultimate inequality follows from the fact that the
makespan of the optimal offline schedule is at most 1.
\end{proof}
\noindent
Finally, we bound the total potential $\phi$; this immediately yields
Lemma~\ref{lma:fractional-main}.
\begin{lemma}
The online fractional algorithm produces a schedule that satisfies 
$\phi = O(m\log m)$.
\end{lemma}
\begin{proof}
Lemmas~\ref{lma:initial}, \ref{lma:category1}, 
\ref{lma:category2} and \ref{lma:category3} imply
$\phi \leq O(m\log m) + (2/3)\phi$,
which proves the lemma.
\end{proof}

\section{The Online Randomized Rounding Procedure} 
\label{sec:rounding}

In this section, we give an online randomized rounding scheme for the fractional solution produced by
the algorithm in the previous section. 

\subsection{The Algorithm}
For each machine $i$, we select (offline) a number $r_i$ uniformly at random and 
independently from $[0, 1]$. In response to a new job $j$ arriving online, the
algorithm updates the schedule in three steps:
\begin{itemize}
\item {\bf Fractional step.} The fractional schedule is updated according to the algorithm 
described in the previous section.
\item {\bf Activation step.} Each inactive machine $i$ that satisfies $r_i \leq 5 x_i(j)\ln (mn)$
is activated.
\item {\bf Assignment step.} Let $M_A(j)$ be the set of active machines in the integer solution. Let 
$z_{ij} = y_{ij}/2 x_i(j)$ if $x_i(j) < 1/5\ln (mn)$ and $z_{ij} = y_{ij}$ otherwise. 
Let $q_{ij}$ be the normalized probability proportional to $z_{ij}$ in a distribution defined on the 
set of machines $M_A(j)$. We assign job $j$ to machine $i$ with probability $q_{ij}$.
\end{itemize}
%

\subsection{Analysis}
%
The next lemma is an immediate consequence of the fact that machine 
$i$ is active in the integer solution with probability
$\min(5x_i\ln (mn), 1)$.
\begin{lemma}
\label{lma:cost-integer}
The total startup cost of all machines activated in the integer
schedule is $O(m \log m\log (mn))$ in expectation.
\end{lemma}
\begin{proof}
The expected startup cost of machine $i$ in the integer schedule 
is at most $5c_i x_i\ln (mn)$;
the lemma now follows from Lemma~\ref{lma:fractional-main} and
linearity of expectation.
\end{proof}
\noindent
To bound the makespan of the integer solution, we first bound the probabilities $q_{ij}$.
(This lemma also shows that with high probability, at least one machine is active for even
one job, thereby proving correctness of the algorithm.)
\begin{lemma}
\label{lma:open}
With probability at least $1 - 1/m$, $q_{ij} \leq z_{ij}$ for all machines $i$ and jobs $j$.
\end{lemma}
\begin{proof}
We show that $\sum_{i\in M_A(j)} z_{ij} \geq 1$ with probability at least $1 - 1/mn$
for any job $j$; the lemma then follows using the union bound over all jobs. 
We classify machines into $M_1(j)$ and $M_2(j)$ depending on whether or not 
$x_i(j) \geq 1/5 \ln (mn)$. Every machine $i\in M_1(j)$ is also in $M_A(j)$. Therefore, 
the contribution of machines in $M_1(j)$
to $\sum_{i\in M_A(j)} z_{ij}$ is exactly $\sum_{i\in M_1(j)} y_{ij}$. On 
the other hand, the contribution of each machine $i\in M_2(j)$ to this sum
is $y_{ij}/2x_i(j)$ with probability $5 x_i(j)\ln (mn)$, and 0 otherwise. Define
a random variable $Z_{ij} = 1$ with probability $5 y_{ij} \ln (mn)/2$, and 0 otherwise. 
Since $y_{ij}/2 x_i(j)\leq 1$,
\begin{equation*}
\pr\left[\sum_{i\in M_2(j) \cap M_A(j)} z_{ij} < \sum_{i\in M_2(j)} y_{ij}\right]
\leq \pr\left[\sum_{i\in M_2(j)} Z_{ij} < \sum_{i\in M_2(j)} y_{ij}\right]
< \frac{1}{mn},
\end{equation*} 
by Chernoff bounds (cf. e.g.~\cite{MotwaniR97}). 
\end{proof}
\noindent
\begin{lemma}
\label{lma:load-whp}
The makespan of the integer schedule is $O(\log m)$ with probability $1 - 2/m$.
\end{lemma}
\begin{proof}
First, we prove that the load on any machine in the integer schedule is 
$O(\log m + \sum_{j\in J} y_{ij} p_{ij})$ with probability at least $1 - 1/m^2$
conditioned on the following:
\begin{itemize} 
\item Lemma~\ref{lma:open} holds, i.e. $q_{ij}\leq z_{ij}$ for all machines $i$ and
jobs $j$, and
\item Machine $i$ is active from the outset in the integer algorithm.
\end{itemize}
Given that machine $i$ is active from the outset and $q_{ij} \leq z_{ij}$, 
the load on machine $i$ due to job $j$ is $p_{ij}$ with probability at most $z_{ij}$. 
Now,
\begin{equation*}
	\sum_{j\in J} z_{ij} p_{ij} 
	= \sum_{j\in J: x_i(j) < 1/5 \ln (mn)} \frac{y_{ij} p_{ij}}{2x_i(j)} + \sum_{j\in J: x_i(j) \geq 1/5 \ln (mn)} y_{ij} p_{ij}
	\leq \sum_{j\in J: x_i(j) < 1} \frac{y_{ij} p_{ij}}{2x_i(j)} + \sum_{j\in J} y_{ij} p_{ij}.
\end{equation*}
Let job $j'$ immediately precede job $j$ in the online order; if $j$ is the first job,
$x_i(j') = 1/m$. Then,
\begin{equation*}
	\sum_{j\in J: x_i(j) < 1} \frac{y_{ij} p_{ij}}{2x_{ij}} 
	\leq 3 \sum_{j\in J: x_i(j) < 1} \frac{x_i(j) - x_i(j')}{x_i(j)}
	\leq 3 \sum_{j\in J: x_i(j) < 1} \int_{w = x_i(j')}^{x_i(j)} \frac{dw}{w}
	\leq 3 \int_{1/m}^1 \frac{dw}{w}
	\leq 3 \ln m.
\end{equation*}
The first inequality follows from the fractional algorithm which assigns a 
fraction $y_{ij} \leq 6(x_i(j) - x_i(j'))/p_{ij}$ of job $j$ to machine $i$.
Since $y_{ij} p_{ij} \leq p_{ij} \leq 1$ for all jobs $j$, it follows using 
Chernoff bounds that the load on machine $i$ is 
$O(\log m + \sum_{j\in J} y_{ij} p_{ij})$ with probability at least $1 - 1/m^2$.

Using the union bound over all machines and Lemma~\ref{lma:open},
we can now claim that the load on machine $i$ is 
$O(\log m + \sum_{j\in J} y_{ij} p_{ij})$ (unconditionally) for all machines $i$,
with probability at least $1-2/m$. 
The lemma now follows using Lemma~\ref{lma:fractional-main}.
\end{proof}
\noindent
Finally, we note that Lemmas~\ref{lma:cost-integer} and \ref{lma:load-whp} imply 
Theorem~\ref{thm:main}.

\newpage

\bibliographystyle{plain}
\bibliography{ref}

\end{document}